\documentclass[11 pt]{article}
\usepackage[centertags]{amsmath}
\usepackage{amsfonts}
\usepackage{amssymb}
\usepackage{amsthm}
\usepackage{newlfont}
\usepackage{xy}
\xyoption{all}
\usepackage{color}
\usepackage{epsfig}
\usepackage{fancyhdr}
\usepackage{graphicx}
\usepackage{float}
\usepackage{lscape}
\usepackage{pdflscape}
\usepackage{multirow}
\usepackage{array}
\usepackage{xcolor}
\usepackage{pdflscape}
\usepackage[margin=1in]{geometry}
\usepackage{hyperref}
\usepackage[nobysame,sorted,citation-order]{amsrefs}

\newtheorem{theorem}{Theorem}[section]
\newtheorem{prop}[theorem]{Proposition}
\newtheorem{cor}[theorem]{Corollary}
\newtheorem{lem}[theorem]{Lemma}

\newcommand{\RR}{\mathbb{R}}

\newcommand{\C}{\mathbb{C}}

\renewcommand{\th}{\text{th}}
\newcommand{\e}{\mathrm{e}}

\newcommand{\al}{\alpha}
\newcommand{\be}{\beta}

\newcommand{\om}{\omega}

\newcommand{\pa}[2]{\dfrac{\partial #1}{\partial #2}}

\newcommand{\lap}{\Delta}

\newcommand{\tot}{\text{tot}}
\newcommand{\DK}{\text{decay}}

\newcommand{\GS}{\text{GS}}
\newcommand{\ratio}{\text{ratio}}

\newcommand{\pnth}[1]{\left( #1 \right)}

\newcommand{\abs}[1]{\left| #1 \right|}

\title{Parametrizations of the Poisson-Schr\"{o}dinger Equations \\ in Spherical Symmetry}
\author{Alan R. Parry\footnote{Mathematics Department, Utah Valley University, MS 261, 800 W. University Parkway, Orem, UT, 84058-6703, email: alan.parry@uvu.edu}}
\date{\today}

\begin{document}

\maketitle

\begin{abstract}
We consider the asymptotically flat standing wave solutions to the Poisson-Schr\"{o}dinger system of equations known as static states. These solutions can be parameterized using a variety of choices of two continuous parameters and one discrete parameter, each having a useful physical-geometrical interpretation. The values of the discrete variable determines the number of nodes (zeros) in the solution. We use numerical inversion techniques to analyze transformations between various informative choices of parametrization by relating each of them to a standard set of three parameters.  Based on our computations, we propose explicit formulas for these relationships. Our computations also show that for the standard choice of continuous variables, the zero-node ground state yields a minimum value of a geometrically natural discrete variable. We give an explicit formula for this minimum value. We use these results to confirm two related observations from previous work by the author and others, and suggest additional applications and approaches to understand these phenomena analytically.
\end{abstract}


\section{Introduction}

The Poisson-Schr\"{o}dinger system of partial differential equations couples the Schr\"{o}dinger equation, which is often used to model the behavior of quantum particles, with the Poisson equation, which models the evolution of a gravitational potential (or really any potential energy field caused by a given density distribution).  It is also the low-field or Newtonian limit of the Einstein-Klein-Gordon system in general relativity \cites{Parry12-2,Goetz15,Giulini12}.  The system therefore has various motivations and applications, many of which stem from the desire to find particles that affect gravity.  This system, which also goes by the name Schr\"{o}dinger-Newton equations due to the Poisson equation being the fundamental equation of Newtonian gravity, is used in quantum mechanics in attempting to quantize gravity \cites{Meter11, Yang13, Bahrami14, Tod99, Moroz98, Ruffini69}.  In general relativity, it has been posed as a model for self-gravitating boson stars, which has been suggested as a model for dark matter \cites{Ruffini69, Bray10, Bray12, Parry12-2, Parry12-3, Goetz14, Goetz15, Lee09, Matos09, MSBS, Seidel90, Seidel98, Bernal08, Sin94, Mielke03, Sharma08, Ji94, Lee92, Lee96, Guzman01, Martinez15, Lopez19}.
Due to these frequently studied applications to this system, this system is very important mathematically and any increase in our current understanding of its solutions is warranted and helpful.  In this paper, we present numerically observed relationships between several parameters that can be used to define the standing wave solutions of the spherically symmetric Poisson-Schr\"{o}dinger equations.  While these relationships are of independent interest, one main hope of this paper is that knowing these relationships will make it easier to compute desired results about solutions to the Poisson-Schr\"{o}dinger equations without necessarily having to numerically compute the solutions from scratch.  Moreover, some of these relationships suggest a method to constrain the value of $m$ for the Poisson-Schr\"{o}dinger model of dark matter.  These numerical observations are an extension of the observations made in the papers \cites{Parry12-3,Goetz15} that investigated some of these relationships between certain parameters of this system.

\subsection{Solutions to the Poisson-Schr\"{o}dinger System}

Given a scalar field $f:\RR^{4} \to \C$, the Poisson-Schr\"{o}dinger equations are the following
\begin{subequations}\label{PSeq}
  \begin{align}
    \label{PSeq-1} \lap_{\RR^{3}}V &= 4\pi (2\abs{f}^{2}) \\
    \label{PSeq-2} i\pa{f}{t} &= -\frac{1}{2m}\lap_{\RR^{3}}f + mVf,
  \end{align}
\end{subequations}
where $V(t,x,y,z)$ is a real valued function and $m$ is a positive constant.  The $\lap_{\RR^{3}}V$ refers to the Laplacian over the $x,y,z$ terms.  Note that we have written the energy density term as $2\abs{f}^{2}$.  This is for consistency with the representation of these equations in the author's previous work.  As an ansatz which helps to find convenient spherically symmetric standing wave solutions, we utilize
\begin{equation}\label{scalarfield}
  f(t,r) = \e^{i(m - \om)t}F(r),
\end{equation}
where $\om$ is a positive constant and $F$ is a real-valued function.  Plugging this into the system~\eqref{PSeq} yields the time-independent Poisson-Schr\"{o}dinger equations
\begin{subequations}\label{SS-PSeq-E}
  \begin{align}
    \label{SS-PSeq-E1} \lap_{\RR^{3}}V &= 4\pi(2F^2) \\
    \label{SS-PSeq-E2} \frac{1}{2m}\lap_{\RR^{3}}F &= (m - \om)F + mVF.
  \end{align}
\end{subequations}
Noting the commonly known fact that in spherical symmetry $\lap_{\RR^{3}}f = f_{rr} + 2f_{r}/r$ and the Newtonian gravity relationship between the mass function $M(r)$ and its corresponding potential function $V(r)$, namely that $V_{r} = M/r^{2}$, the system~\eqref{SS-PSeq-E} becomes
\begin{subequations}\label{SS-PSeq}
  \begin{align}
    \label{SS-PSeq-1} M_{r} &= 8\pi r^{2} F^{2} \\
    \label{SS-PSeq-2} V_{r} &= \frac{M}{r^{2}} \\
    \label{SS-PSeq-3} \frac{1}{2m}\pnth{F_{rr} + \frac{2F_{r}}{r}} &= (m-\om)F + mVF.
  \end{align}
\end{subequations}
Note that these equations are exactly the same as the Newtonian limit of the spherically symmetric Einstein-Klein-Gordon system of equations for which the scalar field in that system is of the form $f(t,r) = \exp(i\om t) F(r)$ \cites{Parry12-2,Goetz15}.  Some useful properties of solutions to this system of ODEs is the topic of this paper.  To get started, we include some general facts about these solutions necessary to our discussion.

First, it is well known that under the boundary conditions that $M \to M_{\tot}$, a constant, (and hence $F \to 0$) and $V \to 0$ as $r \to \infty$, this system produces solutions which are called standing wave solutions or static states where each type of solution can be classified by the number of zeros or nodes $n$ it has \cites{Goetz15,Harrison03,Tod99,Moroz98,Seidel90,Seidel98}.  These solutions oscillate initially, but eventually shift into exponential decay.

Moreover, since $M$ is a mass function and all of the functions are spherically symmetric and smooth, we automatically have that
\begin{align}
  F_{r}(0) &= V_{r}(0) = M_{r}(0)  = 0 &  M(0) &= 0
\end{align}
This leaves the values of $V(0)$, $F(0)$, $m$, $\om$, and $n$ free except for being subject to the constraints $V \to 0$ and $M \to M_{\tot}$ as $r \to \infty$.  It is noted in \cite{Goetz14} that these boundary conditions also require that $0<\om<m$.  Also for the same reasons given in \cite{Goetz14}, we require that $F(0)>0$ and $V(0)<0$.  Following the shooting procedure outlined in \cite{Goetz14}, given a choice of $m$, $\om$, and $n$, we will choose $V(0)$ and $F(0)$ to obtain $n$ zeros and then satisfy $V\to 0$ and $M \to M_{\tot}$ as $r \to \infty$.

It is well-known that we can scale any given solution to the system~\eqref{SS-PSeq} as follows for any $\al,\be >0$
\begin{align}\label{scaling}
    \notag \bar{r} &= \al^{-1}\be^{-1}r & \bar{m} &= \be^{2}m & (\bar{m} - \bar{\om}) &= \al^{2}(m - \om) \\
    \bar{F} &= \al^{2}F & \bar{M} &= \al\be^{-3}M & \bar{V} &= \al^{2}\be^{-2}V
\end{align}
and obtain another solution \cite{Goetz15}.  Thus if we start with a ground state ($n=0$) for specified values of $m$ and $\om$, we can obtain any other ground state for any other values of $m$ and $\om$ restricted as above by using the scalings in equation~\eqref{scaling}.  In preliminary work for \cite{Goetz15}, Goetz computed the ground state through $800^{\th}$ excited state for the values of $m = 100$ and $\om = 99.999$.  This yields a set of solutions that can be scaled as above to yield any solution desired.

Setting up the solutions this way allows for many different ways of parameterizing the solutions.  To determine a solution, one simply needs two conditions which will fix an $\al$ and $\be$ above and then some condition to fix which state, whether ground, first excited, second excited, etc., one desires.  Since all of the values of $r$, $m$, $\om$, etc. on the right hand sides of the scalings in equation~\eqref{scaling} are the values of the set of static states already computed, any choice of two of the values of parameters of these types will fix an $\al$ and $\be$.  Any third parameter can be used to set which type of static state the solutions is, but it cannot be freely chosen.  Instead it will have a value that corresponds to a ground state, one that corresponds to a first excited state, and so on.  That is, there is a discrete set of values of the third parameter that yields a static state.  Thus any three parameters can be used to completely parametrize all static state solutions to the Poisson-Schr\"{o}dinger equations.  Any two of these three can be considered continuous parameters, while the third is then discrete.

\subsection{Parameters Used}\label{par}

We have created a list of useful parameters that can be used for the above process and that have some useful physical interpretation.  Obvious choices include the mass term $m$, the total mass $M_{\tot}$, the value of $\om$, and the number of nodes $n$.  We include several other parameters for a total of 12 parameters and for convenience have listed all of these parameters and their definitions below.
\begin{description}
  \item[[$m$\!\!]] The mass term in the Poisson-Schr\"{o}dinger equations.
  \item[[$M_{\tot}$\!\!]] The total mass of the solution, i.e. $\displaystyle M_{\tot} := \lim_{r \to \infty} M(r)$.
  \item[[$R_{99}$\!\!]] The value of the radius $r$ which contains $99\%$ of the mass.
  \item[[$n$\!\!]] The number of nodes of the static state.
  \item[[$\om$\!\!]] The frequency of the scalar field $F(r)$.
  \item[[$V_{0}$\!\!]] The initial value of the potential function $V(r)$.
  \item[[$F_{0}$\!\!]] The initial value of the scalar field $F(r)$.
  \item[[$R_{h}$\!\!]] The value of $r$ which contains $50\%$ of the mass (i.e. $M(R_{h}) = M_{\tot}/2$).
  \item[[$R_{\DK}$\!\!]] The value of $r$ where the scalar field $F(r)$ shifts from oscillatory behavior to exponential decay behavior.  We call this the decay radius.
  \item[[$M_{\DK}$\!\!]] The mass contained within $r=R_{\DK}$.
  \item[[\! $\abs{F_{\DK}}$\!\!]] The value of $\abs{F(r)}$ at $r=R_{\DK}$.
  \item[[$v_{\max}$\!\!]] The maximum circular orbital velocity of the spacetime solution.
\end{description}
In this paper, we seek to determine how these different parameters are related to each other.  We do this by using three parameters to define the states as described above and then determining how the other nine parameters depend on these three.  Once these relationships are determined, the resulting equations can be used to determine the dependency of any of these parameters on any other.  As a matter of interest, understanding the dependency of each of the defining parameters on the other two turns out to be particularly useful as well.  We describe our specific methods of computing these relationships in section~\ref{methods}.  In section~\ref{results}, we detail the resulting equations.  Finally, in section~\ref{sumoutlook}, we give a summary of all of the equations computed in section~\ref{results} and describe an application where this information might be useful.

\section{Methods}\label{methods}

In this section, we detail how we computed the relationships we will present in section~\ref{results}.  As noted earlier, we choose three parameters as our ``defining'' parameters of the static states and relate the other parameters to these three.  The three defining parameters we choose are $m$, $M_{\tot}$, and $R_{99}$.  We make this choice because they are, in our opinion, the most important physical parameters.  From a quantum physics perspective, the mass term $m$ is the mass of the particle described by the Schr\"{o}dinger equation, $M_{\tot}$ is the total mass of the system, and $R_{99}$ is effectively the total radius of the system.  From a mathematical point of view, it essentially does not matter which three we choose since any of the other relationships can be found by substituting some of the resulting relationships into the others.

We divide the remainder of this section in two.  The first part describes the method we use to determine the dependencies of the other nine parameters on the three parameters we chose as defining parameters, $m$, $M_{\tot}$ and $R_{99}$.  The second part describes, given values for two of the defining parameters, how to compute the minimum value of the third.

\subsection{Parameter Dependencies on Defining Parameters}

The scalings in equation~\eqref{scaling} greatly simplify the task of computing the dependencies of the parameters listed in section~\ref{par} on $m$, $M_{\tot}$, and $R_{99}$.  We will show in the remainder that all of the quantities in section~\ref{par} or closely related ones will be proportional to the product of power functions of these three parameters.  That is, we show that for a quantity $Q$ from section~\ref{par} (or a related quantity), we have that approximately
\begin{equation}\label{rawpower}
  Q = C m^{p}(M_{\tot})^{q}(R_{99})^{s}
\end{equation}
for some constant $C$.  The scalings in equation~\eqref{scaling} restrict what the exponents $p$, $q$, and $s$ can be.  Specifically, if, under the scalings in equation~\eqref{scaling}, the quantity $Q$ scales as
\begin{equation}\label{parscale}
  \bar{Q} = \al^{k}\be^{\ell}Q,
\end{equation}
then equation~\eqref{rawpower} scales as
\begin{align}
  \notag \bar{Q} &= C \bar{m}^{p}(\bar{M}_{\tot})^{q}(\bar{R}_{99})^{s} \\
  \notag \al^{k} \be^{\ell}Q &= C (\be^{2}m)^{p} (\al \be^{-3} M_{\tot})^{q}(\al^{-1}\be^{-1}R_{99})^{s} \\
   &= C\al^{q-s}\be^{2p-3q-s}m^{p}(M_{\tot})^{q}(R_{99})^{s}.
\end{align}
Since any relation of the form in equation~\eqref{rawpower} must preserve the scalings in equation~\eqref{scaling}, the above equation requires that
\begin{subequations}
\begin{align}
  k &= q-s \\
  \ell &= 2p-3q-s.
\end{align}
\end{subequations}
Solving for $p$ and $q$ yields,
\begin{subequations}
\begin{align}
  q &= k + s \\
  p &= \frac{\ell + 3k}{2} + 2s.
\end{align}
\end{subequations}
Plugging this equation in equation~\eqref{rawpower} yields
\begin{equation}\label{refpower}
  Q = C m^{(\ell + 3k)/2}(M_{\tot})^{k}\pnth{m^{2}M_{\tot}R_{99}}^{s}.
\end{equation}
We have essentially proven the following proposition.

\begin{prop}\label{Qscale1}
  If under the scalings in equation~\eqref{scaling}, a quantity $Q$ scales as
   \begin{equation}
     \bar{Q} = \al^{k}\be^{\ell}Q,
   \end{equation}
  then the quantity
   \begin{equation}\label{scaleinv}
     \frac{Q}{m^{(\ell + 3k)/2}(M_{\tot})^{k}(m^{2}M_{\tot}R_{99})^{s}}
   \end{equation}
  is scale invariant under the scalings in equation~\eqref{scaling}.
\end{prop}

This proposition means that, for any value of $s$, the quantity in equation~\eqref{scaleinv} is constant for all solutions of a given state (i.e. a given value of $n$).  The constant itself depends on $n$.  We show, however, that for each of the quantities in section~\ref{par}, there is a value of $s$  such that  the quantity in equation ~\eqref{scaleinv} is appoximately constant across all values of $n$.  To do this, we first note that under the scalings in equation~\eqref{scaling}, the quantity $m^{2}M_{\tot}R_{99}$ satisfies
\begin{equation}
  \bar{m}^{2}\bar{M}_{\tot}\bar{R}_{99} = \be^{4}m^{2}\al\be^{-3}M_{\tot}\al^{-1}\be^{-1}R_{99} = m^{2}M_{\tot}R_{99}
\end{equation}
and hence we have the following lemma.

\begin{lem}\label{defparlem}
  Under the scalings in equation~\eqref{scaling}, the quantity $m^{2}M_{\tot}R_{99}$ is scale invariant.  
\end{lem}

Proposition~\ref{Qscale1} and Lemma~\ref{defparlem} then immediately imply the following corollary.

\begin{cor}\label{Qscale2}
   If under the scalings in equation~\eqref{scaling}, a quantity $Q$ scales as
   \begin{equation}
     \bar{Q} = \al^{k}\be^{\ell}Q,
   \end{equation}
  then the quantity
   \begin{equation}
     \frac{Q}{m^{(\ell + 3k)/2}(M_{\tot})^{k}}
   \end{equation}
  is scale invariant under the scalings in equation~\eqref{scaling}.
\end{cor}

While $k$ and $\ell$ will be known for any given quantity $Q$ from section~\ref{par}, the value of $s$, if it exists, that makes equation~\eqref{scaleinv} constant across all $n$ must be computed numerically.  To determine the power function relationship, if any, between $Q$ and the quantity $m^{2}M_{\tot}R_{99}$, we need to make a log-log plot of 
\begin{equation*}
  \frac{Q}{m^{(\ell + 3k)/2}(M_{\tot})^{k}} \quad \text{against} \quad m^{2}M_{\tot}R_{99}.
\end{equation*}
Since both of these quantities are scale invariant, there is only one data point for each ground or excited state.  Thus, to acquire data points for this plot, we simply need one example of each state.  We use the states computed by Goetz, specifically from the ground to the $800^{\th}$ excited state.  To keep things defined in terms of the defining parameters we chose, we will hold two of the parameters constant (these automatically become the choices for the continuous defining parameters) and vary the third along the values which induce the ground through $800^{\th}$ excited state.  We then compute the resulting log-log plot and its linear regression.  The slope of this best fitting line is the number $s$.  The value of $s$ computed this way is very close to a common rational number in almost every case.  This suggests some relatively simple relationships between these parameters, which we will write down as approximations.

In section~\ref{results}, we compute the best fitting line for each of these log-log plots as well as the correlation coefficient and collect this information in a table for each parameter.  Additionally, these correlations seem to get stronger the higher the excited state, and as such, we will separate these correlations into three regimes: (1) the ground through $20^{\th}$ excited state, (2) the $20^{\th}$ to $100^{\th}$ excited state, and (3) the $100^{\th}$ to $800^{\th}$ excited state.  

As a matter of procedural description, in every case we will hold a $m$ and $M_{\tot}$ constant and vary $R_{99}$.  For the constant values we use the values (all measured in normalized units).
\begin{align}
  m &= 100, &  M_{\tot} &= 0.02.
\end{align}
This is only included for informational purposes.  It really makes no difference what values are used above, the same relationships described in section~\ref{results} will be observed since we are using scale invariant quantities.

\subsection{Minimum Values of Defining Parameters}\label{minvals}

Here we outline the procedure for determining the dependency of the minimum value of each of the parameters $M_{\tot}$, $R_{99}$, and $m$ on the remaining two of these values.  By plotting the values of any one of these three parameters against the number of nodes of the state for constant values of the other two parameters, we see that the minimum value in each case is obtained by the ground state value of that parameter (see Figure~\ref{MinVal}).  This will be the case no matter the choice of value of the other two parameters since the static state solutions are just scalings of a single set.  Thus to accomplish our task of computing the dependency of the minimum value of each of the parameters $M_{\tot}$, $R_{99}$, and $m$ on the remaining two of these values, we simply need to compute the dependency of the ground state value of each of these parameters.  This computation is a relatively straightforward consequence of the scalings in equation~\eqref{scaling} and requires no numerical calculations at all.  For organizational purposes, the result is in section~\ref{mGSsec}.

\begin{figure}[t]

\begin{center}
  \includegraphics[width = 3 in]{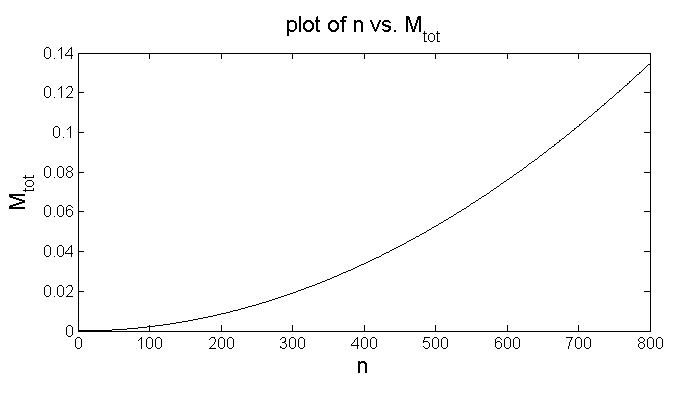}
  \includegraphics[width = 3 in]{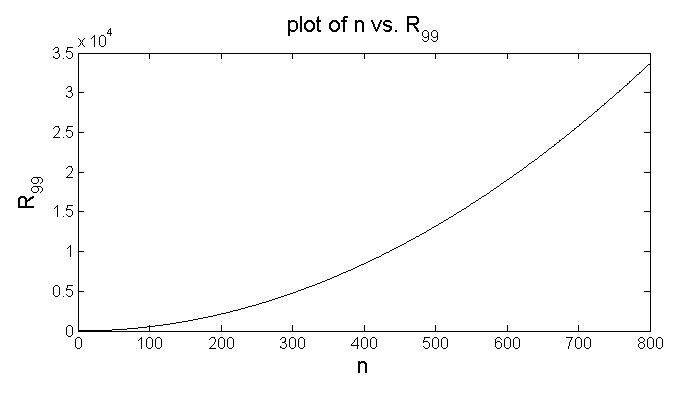}
  \includegraphics[width = 3 in]{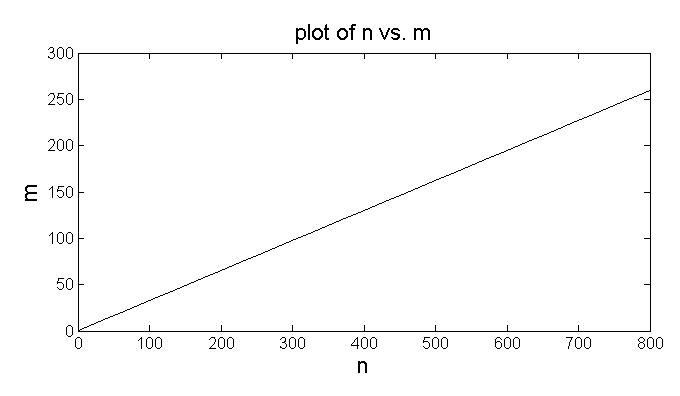}
\end{center}

\caption{Plots of the defining parameters, $M_{\tot}$, $R_{99}$, and $m$, versus the number of nodes $n$.  In the plots, the indicated defining parameter varies while the other two are held constant.  From these, it is clear that the minimum value of each of these parameters, given a fixed value for the other two, occurs for the ground state ($n=0$).  Technically, these plots are discrete as the value of the varying parameter is only defined for each positive integer $n$.}

\label{MinVal}

\end{figure}

\section{Results}\label{results}

In this section, we detail all of the results of the computations described in section~\ref{methods}.  In the first nine subsections, we show how the nine remaining parameters depend on the three parameters we selected as defining parameters, namely, $m$, $M_{\tot}$, and $R_{99}$.  In the remaining section, we show the dependency of the ground state value of each of the defining parameters on the choices of the values of  the other two defining parameters.  This value represents the minimum value that parameter can take given the chosen values of the other two parameters.  All constants in all ten sections have been rounded to four significant digits.

\subsection{Number of Nodes}

\begin{figure}[t]

\begin{center}
  \includegraphics[width = \textwidth]{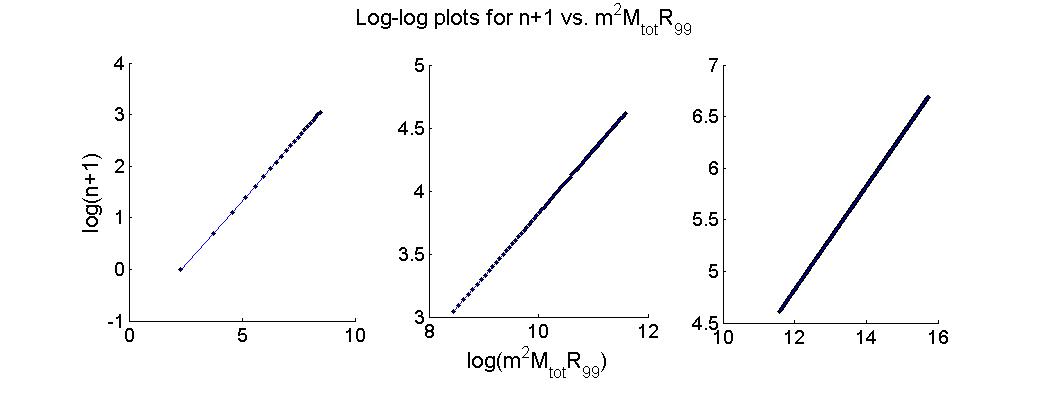}
\end{center}

\caption{Log-log plots of $n+1$ against the quantity $m^{2}M_{\tot}R_{99}$.  The scatter plot in each plot is the actual data computed for the static states, while the line is the best fitting line for the data.  The left plot is for the ground through $20^{\th}$ excited state regime, the center plot for the $20^{\th}$ through $100^{\th}$ excited state, and the right plot for the $100^{\th}$ through $800^{\th}$. The slope and intercept for these lines as well as the correlation coefficient for the data can be found in Table~\ref{statsN}.}

\label{plotsN}

\end{figure}

\begin{table}[ht]

\begin{center}
  \begin{tabular}{c||c|c|c}
    State Range & (0 - 20) & (20 - 100) & (100 - 800)  \\
    \hline Slope ($s$) & 0.4974 & 0.5013 & 0.5004  \\
     Intercept & -1.164  & -1.192 & -1.181  \\
     Exponentiated Intercept ($C$) & 0.3121 & 0.3036 & 0.3070  \\
     Correlation Coefficient & 1.000 & 1.000 & 1.000
  \end{tabular}
\end{center}

\caption{The slope, intercept, and correlation coefficient of the log-log plot of the quantity $n+1$ against the quantity $m^{2}M_{\tot}R_{99}$.  The log-log plots these values come from are found in Figure~\ref{plotsN}.}

\label{statsN}

\end{table}

Due to the fact that the first value of $n$ is $n=0$ denoting the ground state, computing a log-log plot for this parameter would be invalid for this first value.  Thus we consider instead the value $n+1$.  The value of $n+1$ is dimensionless and so will not scale via the scalings in equation~\eqref{scaling}.  This implies that for this parameter $k = \ell = 0$ in equation~\eqref{parscale}.  Thus if $n+1$ is a power function of the defining parameters at all, we must have
\begin{equation}
  n+1 = C \pnth{m^{2}M_{\tot}R_{99}}^{s}
\end{equation}
for some constants $C$ and $s$.  In Table~\ref{statsN}, we detail the numerical observations we made from the log-log plots of $n+1$ versus $\pnth{m^{2}M_{\tot}R_{99}}$ found in Figure~\ref{plotsN}. The best fitting line of this log-log plot suggests that the value of $s$ in this case is approximately $1/2$ and hence it is approximately true that
\begin{equation}\label{nCorr}
  n = C m \sqrt{M_{\tot}R_{99}} - 1.
\end{equation}
where the constant $C$ is the exponentiated intercept given in Table~\ref{statsN} for each of the different regimes.  Note that for large values of $n$, where the $-1$ is less significant, equation~\eqref{nCorr} implies that approximately
\begin{equation}
  n^{2} \propto m^{2}M_{\tot}R_{99},
\end{equation}
which says the same thing as~\eqref{nCorr}, but is written in a more convenient form.

\subsection{Scalar Field Frequency}

\begin{figure}[t]

\begin{center}
  \includegraphics[width = \textwidth]{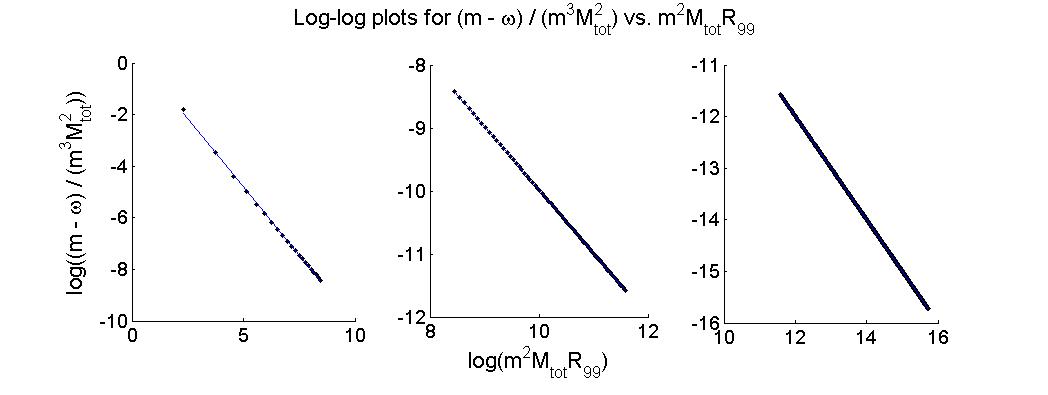}
\end{center}

\caption{Log-log plots of the quantity $(m-\om)/\pnth{m^{3}(M_{\tot})^{2}}$ against the quantity $m^{2}M_{\tot}R_{99}$.  The scatter plot in each plot is the actual data computed for the static states, while the line is the best fitting line for the data.  The left plot is for the ground through $20^{\th}$ excited state regime, the center plot for the $20^{\th}$ through $100^{\th}$ excited state, and the right plot for the $100^{\th}$ through $800^{\th}$.  The slope and intercept for these lines as well as the correlation coefficient for the data can be found in Table~\ref{statsOm}.}

\label{plotsOm}

\end{figure}

\begin{table}[ht]

\begin{center}
  \begin{tabular}{c||c|c|c}
    State Range & (0 - 20) & (20 - 100) & (100 - 800)  \\
    \hline Slope ($s$) & -1.059 & -1.008 & -1.002 \\
     Intercept & 0.4940  & 0.09659 & 0.02435  \\
     Exponentiated Intercept ($C$) & 1.639 & 1.101  & 1.025 \\
     Correlation Coefficient & -0.9997 & -1.000 & -1.000
  \end{tabular}
\end{center}

\caption{The slope, intercept, and correlation coefficient of the log-log plot of the quantity $(m-\om)/\pnth{m^{3}(M_{\tot})^{2}}$ against the quantity $m^{2}M_{\tot}R_{99}$.    The log-log plots these values come from are found in Figure~\ref{plotsOm}.}

\label{statsOm}

\end{table}

The quantity here that appears numerically to satisfy equation~\eqref{rawpower} is $m-\om$ rather than $\om$.  By equation~\eqref{scaling}, $m-\om$ scales such that $k = 2$ and $\ell =0$ in equation~\eqref{parscale}.  Thus we have from equation~\eqref{refpower} that, to be invariant under the scalings, $m-\om$ must satisfy
\begin{equation}
  m-\om = C m^{3/2}(M_{\tot})^{2}\pnth{m^{2}M_{\tot}R_{99}}^{s}
\end{equation}
for some constants $C$ and $s$.  In Table~\ref{statsOm}, we collect the slope and intercept of the best fit line to the log-log plot of $(m-\om)/\pnth{m^{3}(M_{\tot})^{2}}$ versus $m^{2}M_{\tot}R_{99}$.  The plots themselves can be found in Figure~\ref{plotsOm}.  The table suggests that $s \approx -1$ and hence it is approximately true that
\begin{align}
  \notag m - \om &= \frac{CM_{\tot}m}{R_{99}} \\
  \om &= m\pnth{1 - \frac{CM_{\tot}}{R_{99}}}
\end{align}
where the constant $C$ is the exponentiated intercept in Table~\ref{statsOm} for each of the different regimes.

\subsection{Initial Value of the Potential}

\begin{figure}[t]

\begin{center}
  \includegraphics[width = \textwidth]{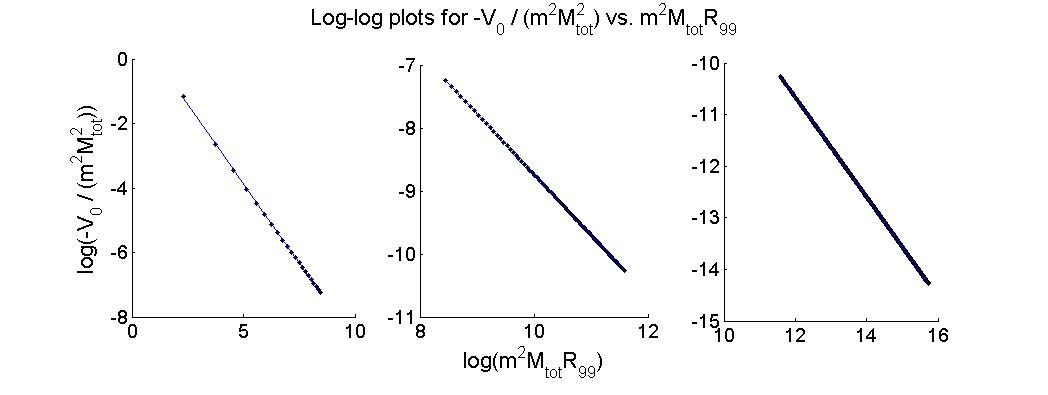}
\end{center}

\caption{Log-log plots of the quantity $-V_{0}/\pnth{m^{2}(M_{\tot})^{2}}$ against the quantity $m^{2}M_{\tot}R_{99}$.  The scatter plot in each plot is the actual data computed for the static states, while the line is the best fitting line for the data.  The left plot is for the ground through $20^{\th}$ excited state regime, the center plot for the $20^{\th}$ through $100^{\th}$ excited state, and the right plot for the $100^{\th}$ through $800^{\th}$.  The slope and intercept for these lines as well as the correlation coefficient for the data can be found in Table~\ref{statsV0}.}

\label{plotsV0}

\end{figure}

\begin{table}[ht]

\begin{center}
  \begin{tabular}{c||c|c|c}
    State Range & (0 - 20) & (20 - 100) & (100 - 800)  \\
    \hline Slope ($s$) & -0.9832 & -0.9633 & -0.9671  \\
     Intercept & 1.047 & 0.8960 & 0.9419  \\
     Exponentiated Intercept ($-C$) & 2.850 & 2.450  & 2.565 \\
     Correlation Coefficient & -0.9999 & -1.000 & -1.000
  \end{tabular}
\end{center}

\caption{The slope, intercept, and correlation coefficient of the log-log plot of the quantity $-V_{0}/\pnth{m^{2}(M_{\tot})^{2}}$ against the quantity $m^{2}M_{\tot}R_{99}$.    The log-log plots these values come from are found in Figure~\ref{plotsV0}.}

\label{statsV0}

\end{table}

From equation~\eqref{scaling}, the initial value of the potential $V_{0}$ scales such that the $k$ and $\ell$ parameters in equation~\eqref{parscale} have the values $k=2$ and $\ell = -2$.  Thus via equation~\eqref{refpower}, if $V_{0}$ is a power function of $m^{2}$, $M_{\tot}$, and $R_{99}$, it must be that
\begin{equation}\label{rawV0}
  V_{0} = C m^{2}(M_{\tot})^{2}\pnth{m^{2}M_{\tot}R_{99}}^{s}
\end{equation}
for some constants $C$ and $s$.  Hence we wish to compute the best fitting line of the log-log plot of $V_{0}/\pnth{m^{2}(M_{\tot})^{2}}$ versus $m^{2}M_{\tot}R_{99}$.  However since $V_{0}<0$, in order to take a logarithm of this value, we have to instead consider its negative.  Thus we actually use the log-log plot of $-V_{0}/\pnth{m^{2}(M_{\tot})^{2}}$ versus $m^{2}M_{\tot}R_{99}$.  The plots for each regime can be found in Figure~\ref{plotsV0}.   The slope and intercept of the best fit lines are collected in Table~\ref{statsV0}.  The values of these parameters suggest that $s \approx -1$ and hence equation~\eqref{rawV0} becomes
\begin{equation}
  V_{0} = \frac{CM_{\tot}}{R_{99}},
\end{equation}
where the value of $C$ is the negative of the exponentiated intercept in Table~\ref{statsV0}.  That is,
\begin{align}
  \notag C_{0-20} &= -\exp(1.047) = -2.850 \\
  \notag C_{20-100} &= -\exp(0.8960) = -2.450 \\
  C_{100-800} &= -\exp(0.9419) = -2.565
\end{align}
where the subscript denotes for which states the constant is valid.

\subsection{Initial Value of the Scalar Field}

\begin{figure}[t]

\begin{center}
  \includegraphics[width = \textwidth]{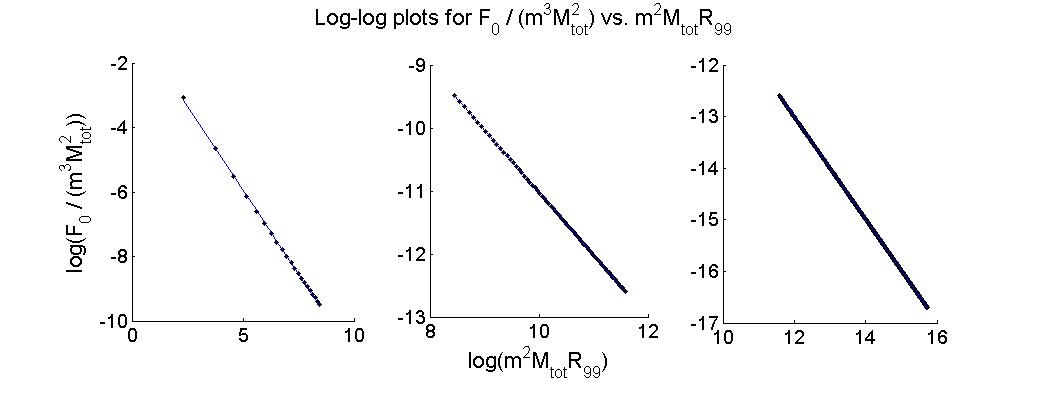}
\end{center}

\caption{Log-log plots of the quantity $F_{0}/\pnth{m^{3}(M_{\tot})^{2}}$ against the quantity $m^{2}M_{\tot}R_{99}$.  The scatter plot in each plot is the actual data computed for the static states, while the line is the best fitting line for the data.  The left plot is for the ground through $20^{\th}$ excited state regime, the center plot for the $20^{\th}$ through $100^{\th}$ excited state, and the right plot for the $100^{\th}$ through $800^{\th}$.  The slope and intercept for these lines as well as the correlation coefficient for the data can be found in Table~\ref{statsF0}.}

\label{plotsF0}

\end{figure}

\begin{table}[ht]

\begin{center}
  \begin{tabular}{c||c|c|c}
    State Range & (0 - 20) & (20 - 100) & (100 - 800)  \\
    \hline Slope ($s$) & -1.033 & -0.9932 &  -0.9907 \\
     Intercept & -0.7844 & -1.092 & -1.120  \\
     Exponentiated Intercept ($C$) & 0.4564 & 0.3356  & 0.3262 \\
     Correlation Coefficient & -0.9998 & -1.000 & -1.000
  \end{tabular}
\end{center}

\caption{The slope, intercept, and correlation coefficient of the log-log plot of the quantity $F_{0}/\pnth{m^{3}(M_{\tot})^{2}}$ against the quantity $m^{2}M_{\tot}R_{99}$.    The log-log plots these values come from are found in Figure~\ref{plotsF0}.}

\label{statsF0}

\end{table}

The initial value of the scalar field $F_{0}$ scales as $F$ does in equation~\eqref{scaling} and hence the $k$ and $\ell$ in equation~\eqref{parscale} are in this case $k=2$ and $\ell = 0$.  Thus for $F_{0}$ to be a power function of the defining parameters, we must have by equation~\eqref{refpower}
\begin{equation}
  F_{0} = C m^{3}(M_{\tot})^{2}\pnth{m^{2}M_{\tot}R_{99}}^{s}
\end{equation}
for some constants $C$ and $s$.  Hence here we compute the log-log plot of $F_{0}/\pnth{m^{3}(M_{\tot})^{2}}$ versus $m^{2}M_{\tot}R_{99}$.  The plot for each regime is displayed in Figure~\ref{plotsF0}, while we collect the slope and intercept of the resulting best fitting line in Table~\ref{statsF0}.  The values in Table~\ref{statsF0} suggest that $s \approx -1$ and hence
\begin{equation}
  F_{0} = \frac{CmM_{\tot}}{R_{99}},
\end{equation}
where the $C$ is given by the exponentiated intercept in Table~\ref{statsF0}.

\subsection{The Half-Mass Radius}

\begin{figure}[t]

\begin{center}
  \includegraphics[width = \textwidth]{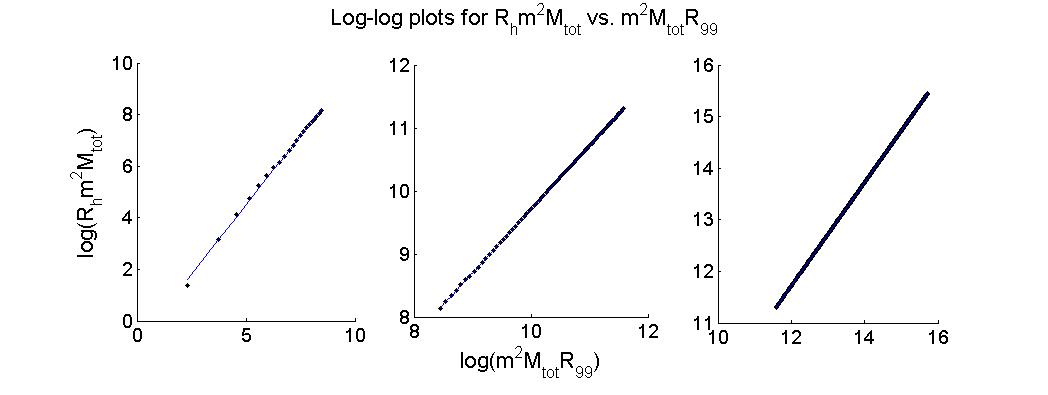}
\end{center}

\caption{Log-log plots of the quantity $R_{h}m^{2}M_{\tot}$ against the quantity $m^{2}M_{\tot}R_{99}$.  The scatter plot in each plot is the actual data computed for the static states, while the line is the best fitting line for the data.  The left plot is for the ground through $20^{\th}$ excited state regime, the center plot for the $20^{\th}$ through $100^{\th}$ excited state, and the right plot for the $100^{\th}$ through $800^{\th}$.  The slope and intercept for these lines as well as the correlation coefficient for the data can be found in Table~\ref{statsRh}.}

\label{plotsRh}

\end{figure}

\begin{table}[ht]

\begin{center}
  \begin{tabular}{c||c|c|c}
    State Range & (0 - 20) & (20 - 100) & (100 - 800)  \\
    \hline Slope ($s$) & 1.073  & 1.007 & 1.002  \\
     Intercept & -0.8665 & -0.3636 &  -0.2959 \\
     Exponentiated Intercept ($C$) & 0.4204 & 0.6951  & 0.7438 \\
     Correlation Coefficient & 0.9991 & 1.000 & 1.000
  \end{tabular}
\end{center}

\caption{The slope, intercept, and correlation coefficient of the log-log plot of the quantity $R_{h}m^{2}M_{\tot}$ against the quantity $m^{2}M_{\tot}R_{99}$.    The log-log plots these values come from are found in Figure~\ref{plotsRh}.}

\label{statsRh}

\end{table}

The half-mass radius $R_{h}$ being a radius will scale like radii in equation~\eqref{scaling}.  Hence for $R_{h}$, the $k$ and $\ell$ in equation~\eqref{parscale} is $k = -1$ and $\ell = -1$.  Thus if $R_{h}$ is a power function of the defining parameters, in order to be invariant under scaling, we must have by equation~\eqref{refpower},
\begin{equation}
  R_{h} = C m^{-2}(M_{\tot})^{-1}\pnth{m^{2}M_{\tot}R_{99}}^{s}
\end{equation}
for some constants $C$ and $s$.  Then we compute the log-log plot of the quantity $R_{h}m^{2}M_{\tot}$ versus the quantity $m^{2}M_{\tot}R_{99}$.  These plots can be found in Figure~\ref{plotsRh}.  The slope and intercept of the best fitting line of this plot can be found in Table~\ref{statsRh}.  The values in that table suggest that $s \approx 1$ and hence
\begin{equation}\label{Rhdep}
  R_{h} = CR_{99}.
\end{equation}
where $C$ is given by the exponentiated intercept in Table~\ref{statsRh}.

\subsection{The Decay Radius}

\begin{figure}[t]

\begin{center}
  \includegraphics[width = \textwidth]{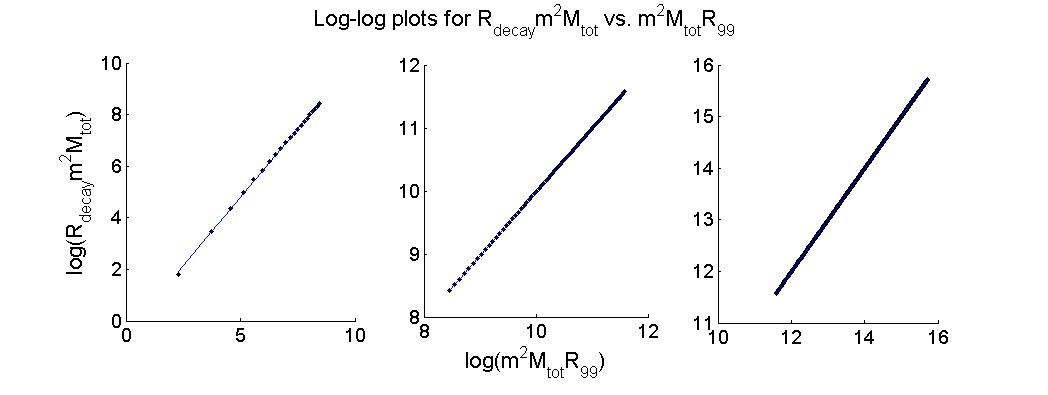}
\end{center}

\caption{Log-log plots of the quantity $R_{\DK}m^{2}M_{\tot}$ against the quantity $m^{2}M_{\tot}R_{99}$.  The scatter plot in each plot is the actual data computed for the static states, while the line is the best fitting line for the data.  The left plot is for the ground through $20^{\th}$ excited state regime, the center plot for the $20^{\th}$ through $100^{\th}$ excited state, and the right plot for the $100^{\th}$ through $800^{\th}$.  The slope and intercept for these lines as well as the correlation coefficient for the data can be found in Table~\ref{statsRDK}.}

\label{plotsRDK}

\end{figure}

\begin{table}[ht]

\begin{center}
  \begin{tabular}{c||c|c|c}
    State Range & (0 - 20) & (20 - 100) & (100 - 800)  \\
    \hline Slope ($s$) & 1.062 & 1.008 & 1.002  \\
     Intercept & -0.5225 & -0.09945 & -0.02490  \\
     Exponentiated Intercept ($C$) & 0.5930 & 0.9053 & 0.9754 \\
     Correlation Coefficient & 0.9997 & 1.000 & 1.000
  \end{tabular}
\end{center}

\caption{The slope, intercept, and correlation coefficient of the log-log plot of the quantity $R_{\DK}m^{2}M_{\tot}$ against the quantity $m^{2}M_{\tot}R_{99}$.    The log-log plots these values come from are found in Figure~\ref{plotsRDK}.}

\label{statsRDK}

\end{table}

Similar to the half-mass radius, the decay radius $R_{\DK}$ scales as equation~\eqref{parscale} for $k = -1$ and $\ell = -1$ and hence if $R_{\DK}$ is to be a power function of the defining parameters and remain invariant under scaling, by equation~\eqref{refpower}, we must have
\begin{equation}
  R_{\DK} = C m^{-2}(M_{\tot})^{-1}\pnth{m^{2}M_{\tot}R_{99}}^{s}
\end{equation}
for some constants $C$ and $s$.  Thus we compute the log-log plot of the quantity $R_{\DK}m^{2}M_{\tot}$ versus the quantity $m^{2}M_{\tot}R_{99}$.  These plots are found in Figure~\ref{plotsRDK}, while the slope and intercept of the best fitting line of this plot can be found in Table~\ref{statsRDK}.  The values in that table suggest that $s \approx 1$ and hence
\begin{equation}\label{RDKCorr}
  R_{\DK} = CR_{99},
\end{equation}
where $C$ is given by the exponentiated intercept in Table~\ref{statsRDK}.  This suggests that $R_{\DK}$ approaches $R_{99}$ for higher excited states.

\subsection{The Mass Contained Within the Decay Radius}

\begin{figure}[t]

\begin{center}
  \includegraphics[width = \textwidth]{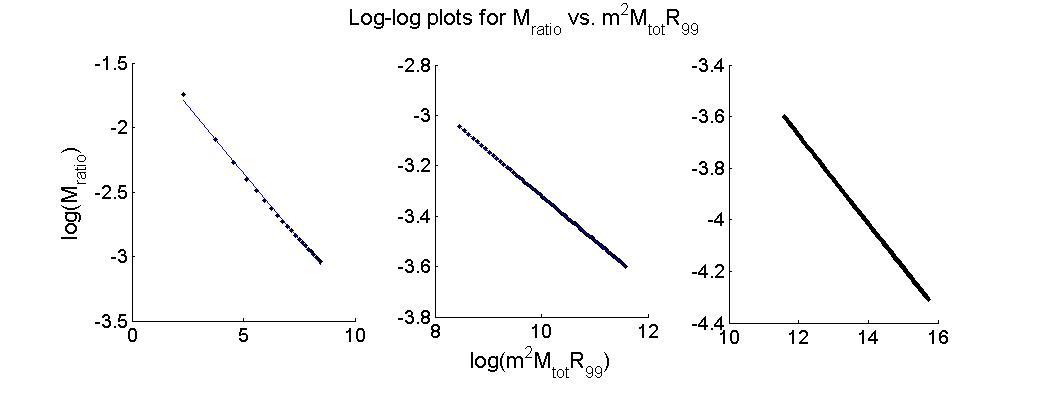}
\end{center}

\caption{Log-log plots of the quantity $M_{\ratio}$ against the quantity $m^{2}M_{\tot}R_{99}$.  The scatter plot in each plot is the actual data computed for the static states, while the line is the best fitting line for the data.  The left plot is for the ground through $20^{\th}$ excited state regime, the center plot for the $20^{\th}$ through $100^{\th}$ excited state, and the right plot for the $100^{\th}$ through $800^{\th}$.  The slope and intercept for these lines as well as the correlation coefficient for the data can be found in Table~\ref{statsMDK}.}

\label{plotsMDK}

\end{figure}

\begin{table}[ht]

\begin{center}
  \begin{tabular}{c||c|c|c}
    State Range & (0 - 20) & (20 - 100) & (100 - 800)  \\
    \hline Slope ($s$) & -0.2059  & -0.1770 & -0.1710  \\
     Intercept & -1.321 & -1.550 & -1.621  \\
     Exponentiated Intercept ($C$) & 0.2669 & 0.2123  & 0.1976 \\
     Correlation Coefficient & -0.9987 & -0.9999 & -1.000
  \end{tabular}
\end{center}

\caption{The slope, intercept, and correlation coefficient of the log-log plot of the quantity $M_{\ratio}$ against the quantity $m^{2}M_{\tot}R_{99}$.  The log-log plots these values come from are found in Figure~\ref{plotsMDK}.}

\label{statsMDK}

\end{table}

The value of $M_{\DK}$ and $M_{\tot}$ are so close to each other that the log-log plots for $M_{\DK}$ and the defining parameters give misleading results.  In order to focus in on the difference between $M_{\DK}$ and $M_{\tot}$, we consider the alternative quantity
\begin{equation}\label{massratio}
  M_{\ratio} = 1 - \frac{M_{\DK}}{M_{\tot}}
\end{equation}
and compute the dependency of this quantity on the defining parameters, $M_{\tot}$, $R_{99}$, and $m$.  Note that since $M_{\DK}$ and $M_{\tot}$ both scale the same way, the quantity $M_{\ratio}$ is dimensionless and hence invariant under the scalings in equation~\eqref{scaling}.  This implies that for this quantity, $k=0$ and $\ell = 0$ in equation~\eqref{parscale}.  Thus if $M_{\ratio}$ is a power function of the defining parameters and remains invariant under scaling, by equation~\eqref{refpower} it must be that
\begin{equation}
  M_{\ratio} = C \pnth{m^{2}M_{\tot}R_{99}}^{s}
\end{equation}
for some constants $C$ and $s$.  Thus we compute the log-log plot of $M_{\ratio}$ versus $m^{2} M_{\tot}R_{99}$.  These plots are displayed in Figure~\ref{plotsMDK}.  Table~\ref{statsMDK} contains the slope and intercept of the best fitting line to this plot and suggests that in this case $s \approx -0.175 = -7/40$.  This yields that approximately
\begin{subequations}
\begin{align}
  M_{\ratio} = 1 - \frac{M_{\DK}}{M_{\tot}} &= \frac{C}{(m^{2}R_{99}M_{\tot})^{(7/40)}} \\
  \label{massdecay} M_{\DK} &= \pnth{1 -  \frac{C}{(m^{2}R_{99}M_{\tot})^{(7/40)}}}M_{\tot}  
\end{align}
\end{subequations}
where the constant $C$ is the exponentiated intercept $C$ in Table~\ref{statsMDK}.  Note that since $C$, $m$, $R_{99}$, and $M_{\tot}$ are all positive, equation~\eqref{massdecay} implies that $M_{\DK} < M_{\tot}$ which is obviously a neccessary requirement and serves as a reality check to this equation.  Also note that by equation~\eqref{nCorr}, equation~\eqref{massdecay} can be rewritten as
\begin{equation}\label{MbyN}
  M_{\DK} = \pnth{1 - \frac{\hat{C}}{(n+1)^{(7/20)}}}M_{\tot}
\end{equation}
where $\hat{C}$ is ratio of the $C$ from equation~\eqref{massdecay} and the $C$ from equation~\eqref{nCorr} raised to the power of $7/20$.  Thus if this approximation continues to hold as expected for all static states, it is clear from equation~\eqref{MbyN} that as $n \to \infty$, $M_{\DK} \to M_{\tot}$.

\subsection{The Value of the Scalar Field at the Decay Radius}

\begin{figure}[t]

\begin{center}
  \includegraphics[width = \textwidth]{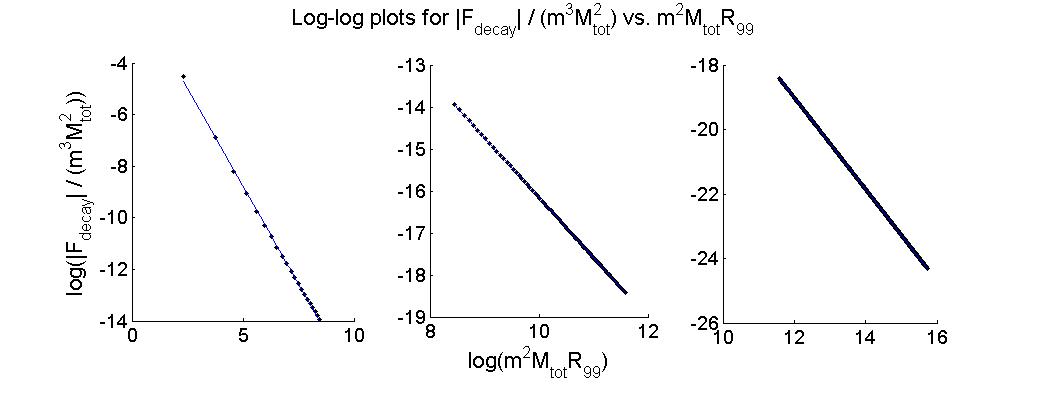}
\end{center}

\caption{Log-log plots of the quantity $\abs{F_{\DK}}/\pnth{m^{3}(M_{\tot})^{2}}$ against the quantity $m^{2}M_{\tot}R_{99}$.  The scatter plot in each plot is the actual data computed for the static states, while the line is the best fitting line for the data.  The left plot is for the ground through $20^{\th}$ excited state regime, the center plot for the $20^{\th}$ through $100^{\th}$ excited state, and the right plot for the $100^{\th}$ through $800^{\th}$.  The slope and intercept for these lines as well as the correlation coefficient for the data can be found in Table~\ref{statsFDK}.}

\label{plotsFDK}

\end{figure}

\begin{table}[ht]

\begin{center}
  \begin{tabular}{c||c|c|c}
    State Range & (0 - 20) & (20 - 100) & (100 - 800)  \\
    \hline Slope ($s$) & -1.509 & -1.430  & -1.420  \\
     Intercept & -1.238 & -1.851 & -1.972  \\
     Exponentiated Intercept ($C$) & 0.2900  & 0.1571  & 0.1392 \\
     Correlation Coefficient & -0.9997 & -1.000 & -1.000
  \end{tabular}
\end{center}

\caption{The slope, intercept, and correlation coefficient of the log-log plot of the quantity $\abs{F_{\DK}}/\pnth{m^{3}(M_{\tot})^{2}}$ against the quantity $m^{2}M_{\tot}R_{99}$.  The log-log plots these values come from are found in Figure~\ref{plotsFDK}.}

\label{statsFDK}

\end{table}

Since the value of the scalar field $F$ at $r = R_{\DK}$ can be either positive or negative depending on whether or not it is computed for an even or odd state respectively, we computed how $\abs{F(R_{\DK})}~=~\abs{F_{\DK}}$ depends on $M_{\tot}$, $R_{99}$, and $m$.  This value will scale like $F$ does in equation~\eqref{scaling} and hence $k=2$ and $\ell = 0$ for this quantity in equation~\eqref{parscale}.  Thus if this quantity is a power function of the defining parameters that is invariant under the scalings, by equation~\eqref{refpower}, it must be that
\begin{equation}\label{FDKraw}
  \abs{F_{\DK}} = C m^{3}(M_{\tot})^{2}\pnth{m^{2}M_{\tot}R_{99}}^{s}
\end{equation}
for some constants $C$ and $s$.  We then compute the log-log plot of $\abs{F_{\DK}}/\pnth{m^{3}(M_{\tot})^{2}}$ versus $m^{2}M_{\tot}R_{99}$.  The plots for each regime are displayed in Figure~\ref{plotsFDK}.  Table~\ref{statsFDK} details the resulting slope and intercept of the best fitting line to this plot.  The slopes here suggest that $s \approx -1.420$.  To guide our guess at what fraction this might be close to, we refer to Goetz's paper \cite{Goetz15}.  Recall that for large values of the number of nodes $n$, the value of $C$ in equation~\eqref{RDKCorr} suggested that $R_{\DK} \approx R_{99}$.  In \cite{Goetz15}, Goetz presented a theoretical argument that for large values of the number of nodes $n$, equation~\eqref{FDKraw} (with $R_{99}$ replaced with $R_{\DK}$) should hold for $s = -17/12$.  And indeed, this is consistent with the observation here as $-17/12 \approx -1.417$.  As such, we use the value of $-17/12$ for $s$ here and obtain approximately
\begin{equation}\label{FDKCorr}
  \abs{F_{\DK}} = \frac{C m^{(1/6)}(M_{\tot})^{(7/12)}}{(R_{99})^{(17/12)}}
\end{equation}
where the value of $C$ is given by the exponentiated intercept in Table~\ref{statsFDK}.  Equation~\eqref{FDKCorr} is precisely the relationship conjectured by Goetz in \cite{Goetz15} and in fact the value of $C$ Goetz obtains is approximately $0.133$ which is consistent with our observation of a value of $0.1392$ in the $100^{\th}$ to $800^{\th}$ excited state regime.

\subsection{The Maximum Circular Orbital Velocity}

\begin{figure}[t]

\begin{center}
  \includegraphics[width = \textwidth]{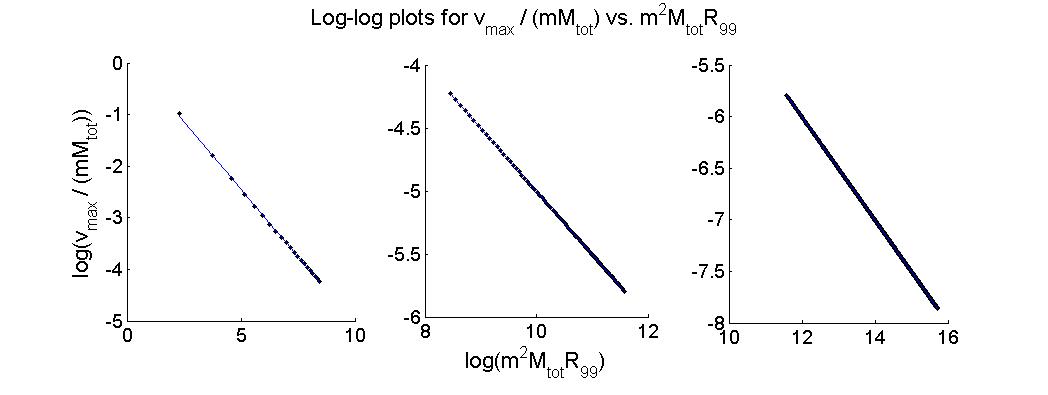}
\end{center}

\caption{Log-log plots of the quantity $v_{\max}/(mM_{\tot})$ against the quantity $m^{2}M_{\tot}R_{99}$.  The scatter plot in each plot is the actual data computed for the static states, while the line is the best fitting line for the data.  The left plot is for the ground through $20^{\th}$ excited state regime, the center plot for the $20^{\th}$ through $100^{\th}$ excited state, and the right plot for the $100^{\th}$ through $800^{\th}$.  The slope and intercept for these lines as well as the correlation coefficient for the data can be found in Table~\ref{statsVmax}.}

\label{plotsVmax}

\end{figure}

\begin{table}[ht]

\begin{center}
  \begin{tabular}{c||c|c|c}
    State Range & (0 - 20) & (20 - 100) & (100 - 800)  \\
    \hline Slope ($s$) & -0.5209 & -0.5004 & -0.4994  \\
     Intercept & 0.1573 & $1.237 \times 10^{-4}$ & -0.01138  \\
     Exponentiated Intercept ($C$) & 1.1704 & 1.000  &  0.9887 \\
     Correlation Coefficient & -0.9998 & -1.000 & -1.000
  \end{tabular}
\end{center}

\caption{The slope, intercept, and correlation coefficient of the log-log plot of the quantity $v_{\max}/(mM_{\tot})$ against the quantity $m^{2}M_{\tot}R_{99}$.  The log-log plots these values come from are found in Figure~\ref{plotsVmax}.}

\label{statsVmax}

\end{table}

The value of $v_{\max}$ is computed as the maximum of the function
\begin{equation}\label{vdef}
  v(r) = \sqrt{\frac{M(r)}{r}}
\end{equation}
and so should scale as
\begin{equation}
  \bar{v} = \sqrt{\frac{\bar{M}}{\bar{r}}} = \sqrt{\frac{\al \be^{-3}M}{\al^{-1}\be^{-1}r}} = \al\be^{-1}\sqrt{\frac{M}{r}} = \al\be^{-1}v
\end{equation}
and hence for this quantity $k = 1$ and $\ell = -1$ in equation~\eqref{parscale}.  Thus if $v_{\max}$ is to be a power function of the defining parameters and invariant under this scaling, by equation~\eqref{refpower}, we must have that
\begin{equation}
  v_{\max} = C m M_{\tot} \pnth{m^{2}M_{\tot}R_{99}}^{s}
\end{equation}
for constants $C$ and $s$.  Thus we compute the log-log plot of $v_{\max}/(mM_{\tot})$ versus $m^{2}M_{\tot}R_{99}$.  These plots are displayed in Figure~\ref{plotsVmax}.  In Table~\ref{statsVmax}, we collect the slope and intercept of the best fitting line of this plot.  This shows that in this case $s \approx -1/2$, which suggests that approximately
\begin{equation}\label{vmaxdep}
  v_{\max} = C\sqrt{\frac{M_{\tot}}{R_{99}}}
\end{equation}
where $C$ is given by the exponentiated intercept in Table~\ref{statsVmax}.  This relationship is not that surprising given that this value is the maximum of the function defined in equation~\eqref{vdef}.

\subsection{Fixed Excited State Values of the Defining Parameters}\label{mGSsec}

It was observed in section~\ref{minvals} that for fixed values of any pair of parameters from $m$, $M_{\tot}$ and $R_{99}$, the minimum value of the third parameter is obtained by the ground state solution.  Thus to determine the minimum values of the parameter, we need only describe how these values are related for a ground state.  This is quite simple due to the scalings in equation~\eqref{scaling} and in fact we can say even more.  This is contained in the following proposition.

\begin{prop}\label{defparprop}
  For a specified ground or excited state (i.e. for a fixed $n$), we have that
    \begin{equation}
      m = \frac{C}{\sqrt{M_{\tot}R_{99}}}
    \end{equation}
  or equivalently
    \begin{equation}\label{hyperbola}
      M_{\tot}R_{99} = \frac{C^{2}}{m^{2}}
    \end{equation}
    for some constant $C$ which depends only on $n$.
\end{prop}

\begin{proof}
  If a quantity is scale invariant under the scalings in equation~\eqref{scaling}, then for a fixed $n$, that quantity equals a constant no matter how one scales the solution.  By Lemma~\ref{defparlem}, $m^{2}M_{\tot}R_{99}$ is scale invariant (and positive) and hence for a fixed $n$,
  \begin{align}
    \notag m^{2}M_{\tot}R_{99} &= C^{\ast} \\
    m &= \frac{C}{\sqrt{M_{\tot}R_{99}}}
\end{align}
where $C = \sqrt{C^{\ast}}$.
\end{proof}

In particular, Proposition~\ref{defparprop} is true for the ground state which yields the minimum value of any one of $m$, $M_{\tot}$, or $R_{99}$ for fixed values of the other two.  Computing the quantity $m\sqrt{M_{\tot}R_{99}}$ for the ground state in the standard set of 800 states that we have been scaling, yields that 
\begin{equation}
  C = 3.155.
\end{equation}
Thus the minimum value of any of these parameters for fixed values of the other two parameters is given approximately by the equation
\begin{equation}\label{mGSdep}
  m = \frac{3.155}{\sqrt{M_{\tot}R_{99}}}.
\end{equation}

Note that equation~\eqref{hyperbola} implies that the mass profile of a static state lies along a hyperbola of constant $m$.  This observation was already made by the author \cite{Parry12-3} and was used by Bray and Parry \cite{Parry12-2} to describe a procedure for finding an upper bound on $m$ in the wave dark matter model.  Note that in \cite{Parry12-3}, the Einstein-Klein-Gordon equations were used instead of the Poisson-Schr\"{o}dinger equations.  Since all of the observations in that paper were made in the low field limit, this is entirely consistent with the observation here since the Poisson-Schr\"{o}dinger equations are the low-field limit version of the Einstein-Klein-Gordon equations.

To complete this consistency comparison, we compute the constant here that is comparable to the observation made by Parry.  In \cite{Parry12-2}, Parry compares the value of $m$ with the value of the total mass $M_{\tot}$ and the half mass radius, $R_{h}$.  We have two of those parameters in equation~\eqref{hyperbola}, but need to substitute in equation~\eqref{Rhdep}.  Since all of Parry's observations were in the regime of ground state to $20^{\th}$ excited state, we will use that value of the constant in equation~\eqref{Rhdep}.  This yields the equation
\begin{equation}
  M_{\tot}R_{h} = \frac{C^{\ast}}{m^{2}}
\end{equation}
where, by equations~\eqref{Rhdep} and~\eqref{mGSdep}, $C^{\ast}$ is the product of $3.155^{2}$ and the value of $C$ in equation~\eqref{Rhdep} for the ground through $20^{\th}$ excited state regime.  This yields that $C^{\ast} \approx 4.185$.  The value of $C^{\ast}$ found in \cite{Parry12-3} is about $3.86$, so these are relatively close.  The small discrepancy is due to the fact that in substituting in $R_{h}$, we used a fit for the ground through $20^{\th}$ excited state while the constant from \cite{Parry12-3} only dealt with the ground state.  The value of $R_{h}/R_{99}$ for the ground state individually is actually only about $0.3945$ which yields a $C^{\ast}$ value of $3.926$ which is on the edge of the error bars listed in \cite{Parry12-3} for the constant there corresponding to $C^{\ast}$.  Thus our observations here are consistent with those in \cite{Parry12-3}.

\section{Summary and Outlook}\label{sumoutlook}

In this section, for convenience, we relist the relationships found in the previous section.  Then we describe a possible application of this work.

\subsection{Summary of Equations}\label{sumeq}

Here we collect all of the observed relationships of the different parameters describing the spherically symmetric static state solutions of the Poisson-Schr\"{o}dinger equations.  We note that these equations are approximate and they appear to work best for higher excited states (e.g. the $100^{\th}$~-~$800^{\th}$ excited state regime).
\begin{subequations}\label{relations}
\begin{align}
  n &= C m \sqrt{M_{\tot}R_{99}}           &               \om &= m\pnth{1 - \frac{CM_{\tot}}{R_{99}}} \\
  V_{0} &= \frac{CM_{\tot}}{R_{99}}     &               F_{0} &= \frac{CmM_{\tot}}{R_{99}} \\
  R_{h} &= CR_{99}			  &	        R_{\DK} &= CR_{99} \\
  M_{\DK} &= \pnth{1 -  \frac{C}{(m^{2}R_{99}M_{\tot})^{(7/40)}}}M_{\tot}     
                                                                  &      \abs{F_{\DK}} &= \frac{C m^{(1/6)}(M_{\tot})^{(7/12)}}{(R_{99})^{(17/12)}} \\
  v_{\max} &= C\sqrt{\frac{M_{\tot}}{R_{99}}}  & m^{\GS} &= \frac{3.155}{\sqrt{M_{\tot}^{\GS}R_{99}^{\GS}}}
\end{align}
\end{subequations}
where $m^{\GS}$, $M_{\tot}^{\GS}$, and $R_{99}^{\GS}$ are the values of these parameters for any ground state and the constants $C$ are collected in Table~\ref{constants}.

\begin{table}[t]

\begin{center}

\begin{tabular}{c|c|c|c|c|c|c|c|c|c}
  State Range & $n$  &  $\om$  &  $V_{0}$  &  $F_{0}$  &  $R_{h}$  &  $R_{\DK}$  &  $M_{\DK}$  &  $\abs{F_{\DK}}$  &  $v_{\max}$ \\
  \hline (0 - 20) & 0.3121 & 1.639 & -2.850 & 0.4564 & 0.4204 & 0.5930 & 0.2669 & 0.2900 & 1.1704 \\
  (20 - 100) & 0.3036  & 1.101 & -2.450 & 0.3356 & 0.6951 & 0.9053 & 0.2123 & 0.1571 & 1.000 \\
  (100 - 800) & 0.3070  & 1.025  & -2.565 & 0.3262 & 0.7438 & 0.9754 & 0.1976 & 0.1392 & 0.9887 
\end{tabular}

\end{center}

\caption{Constant values for $C$ in equation~\eqref{relations}.}
\label{constants}

\end{table}

\subsection{Outlook}\label{outlook}

An interesting open problem with regard to the proposed relationships in equation~\eqref{relations} is to verify these suggestions analytically.  All the relationships in equation~\eqref{relations} are educated guesses based off the numerical data in this paper, but we do not, in this paper, attempt to prove these relationships analytically.  Goetz attempted some justification for Equation~\eqref{FDKCorr} in \cite{Goetz15}, but an actual proof of these relationships would require a more thorough investigation.  Thus one could take equation~\eqref{relations} as a conjecture and work to prove that these relations are true or correct them, if necessary.

A possible application of these observations utilizes equation~\eqref{mGSdep}.  Since $m$ in equation~\eqref{mGSdep} represents the minimum value of $m$ for given values of $M_{\tot}$ and $R_{99}$, it is possible to use~\eqref{mGSdep} to get a lower bound for $m$ in the scalar field dark matter model.  If, as is assumed in \cites{Parry12-2,Goetz14,Goetz15} and several other references, the dark matter halo can be modeled by a static state scalar field, then a lower bound on $m$ could be found by taking a survey of as many dark matter halos as information about their total mass and radius is known or at least approximated, then using~\eqref{mGSdep} to compute $m$ for the ground state in each case, and finally taking the largest of these values as the lower bound for $m$.  The largest value of these minimums is the right value since any value of $m$ smaller than this would not produce a static state solution for at least one halo.

Another possible application has to do with B\"{a}cklund transformations.  Equation~\eqref{nCorr} suggests a relationship between the number of nodes in a state and the standard set of defining parameters of the system.  For soliton solutions to the sine-Gordon equations, B\"{a}cklund transformations successively applied to the 1-soliton solution can be used to find an $n$-soliton solution \cites{Rogers02}.  We suggest the observations in this paper might be used to discover a B\"{a}cklund-like transformation for the standing wave solutions of the Poisson-Schr\"{o}dinger equations that would relate solutions for the ground state to an $n^{\text{th}}$ excited state.  Using different techniques, a similar course of action is suggested by Fuentes et al.\ in \cites{Fuentes17}. 

Due to the above applications and open problems and the vast number of other applications of the spherically symmetric Poisson-Schr\"{o}dinger equations currently studied, the results of the paper should be useful as they allow researchers to quickly compute approximate values about static state solutions to the Poisson-Schr\"{o}dinger equations without needing to take the time to solve for them numerically themselves.

\section{Acknowledgements}

The author gratefully acknowledges the invaluable contributions of and discussions with Andrew S. Goetz in the preparation of this paper.

\bibliographystyle{amsalpha}
\bibliography{./References}

\end{document}